\documentclass{wscpaperproc}
\usepackage{latexsym}
\usepackage{graphicx}
\usepackage{mathptmx}
\usepackage[T1]{fontenc}
\usepackage{xspace}
\usepackage{wrapfig}

\usepackage{amsmath}
\usepackage{amsfonts}
\usepackage{amssymb}
\usepackage{amsbsy}
\usepackage{amsthm}
\usepackage{cite}
\usepackage{soul}
\usepackage{amsmath,amssymb,amsfonts}
\usepackage{algorithmic}
\usepackage[linesnumbered,ruled,vlined]{algorithm2e}
\usepackage{graphicx}
\usepackage{textcomp}
\usepackage{xcolor}
\usepackage{newtxmath}

\usepackage[pdftex,colorlinks=true,urlcolor=blue,citecolor=black,anchorcolor=black,linkcolor=black]{hyperref}

\newtheoremstyle{wsc}
{3pt}
{3pt}
{}
{}
{\bf}
{}
{.5em}
{}

\theoremstyle{wsc}
\newtheorem{theorem}{Theorem}

\newtheorem{definition}{Definition}

\begin{document}

%
%

\pagestyle{fancyplain}

\thispagestyle{plain}
\firstPageHead{}

\chead{\fancyplain{}{\itshape Chen, Banerjee, George, Busart}}

\rhead{}
\cfoot{}
\renewcommand{\headrulewidth}{0pt} 

\makeatletter
\let\@internalcite\cite
\def\cite{\def\@citeseppen{-1000}%
    \def\@cite##1##2{(##1\if@tempswa , ##2\fi)}%
    \def\citeauthoryear##1##2##3{##1 ##3}\@internalcite}
\def\citeNP{\def\@citeseppen{-1000}%
    \def\@cite##1##2{##1\if@tempswa , ##2\fi}%
    \def\citeauthoryear##1##2##3{##1 ##3}\@internalcite}
\def\citeN{\def\@citeseppen{-1000}%
    \def\@cite##1##2{##1\if@tempswa, ##2)\else{}\fi}%
    \def\citeauthoryear##1##2##3{##1 (##3)}\@citedata}
\def\citeA{\def\@citeseppen{-1000}%
    \def\@cite##1##2{(##1\if@tempswa , ##2\fi)}%
    \def\citeauthoryear##1##2##3{##1}\@internalcite}
\def\citeANP{\def\@citeseppen{-1000}%
    \def\@cite##1##2{##1\if@tempswa , ##2\fi}%
    \def\citeauthoryear##1##2##3{##1}\@internalcite}
\def\shortcite{\def\@citeseppen{-1000}%
    \def\@cite##1##2{(##1\if@tempswa , ##2\fi)}%
    \def\citeauthoryear##1##2##3{##2 ##3}\@internalcite}
\def\shortciteNP{\def\@citeseppen{-1000}%
    \def\@cite##1##2{##1\if@tempswa , ##2\fi}%
    \def\citeauthoryear##1##2##3{##2 ##3}\@internalcite}
\def\shortciteN{\def\@citeseppen{-1000}%
    \def\@cite##1##2{##1\if@tempswa, ##2\else{}\fi}%
    \def\citeauthoryear##1##2##3{##2 (##3)}\@citedata}
\def\shortciteA{\def\@citeseppen{-1000}%
    \def\@cite##1##2{(##1\if@tempswa , ##2\fi)}%
    \def\citeauthoryear##1##2##3{##2}\@internalcite}
\def\shortciteANP{\def\@citeseppen{-1000}%
    \def\@cite##1##2{##1\if@tempswa , ##2\fi}%
    \def\citeauthoryear##1##2##3{##2}\@internalcite}
\def\citeyear{\def\@citeseppen{-1000}%
    \def\@cite##1##2{(##1\if@tempswa , ##2\fi)}%
    \def\citeauthoryear##1##2##3{##3}\@citedata}
\def\citeyearNP{\def\@citeseppen{-1000}%
    \def\@cite##1##2{##1\if@tempswa , ##2\fi}%
    \def\citeauthoryear##1##2##3{##3}\@citedata}
%
%
%
\def\@citedata{%
    \@ifnextchar [{\@tempswatrue\@citedatax}%
                  {\@tempswafalse\@citedatax[]}%
}

\def\@citedatax[#1]#2{%
\if@filesw\immediate\write\@auxout{\string\citation{#2}}\fi%
  \def\@citea{}\@cite{\@for\@citeb:=#2\do%
    {\@citea\def\@citea{, }\@ifundefined
       {b@\@citeb}{{\bf ?}%
       \@warning{Citation `\@citeb' on page \thepage \space undefined}}%
{\csname b@\@citeb\endcsname}}}{#1}}%

%
\def\@citex[#1]#2{%
\if@filesw\immediate\write\@auxout{\string\citation{#2}}\fi%
  \def\@citea{}\@cite{\@for\@citeb:=#2\do%
    {\@citea\def\@citea{; }\@ifundefined
       {b@\@citeb}{{\bf ?}%
       \@warning{Citation `\@citeb' on page \thepage \space undefined}}%
{\csname b@\@citeb\endcsname}}}{#1}}%

%
\def\@biblabel#1{}
\makeatother



\newdimen\bibindent
\bibindent=0.0em
\def\thebibliography#1{\section*{\refname}\list
   {}{\settowidth\labelwidth{[#1]}
   \leftmargin\parindent
   \itemindent -\parindent
   \listparindent \itemindent
   \itemsep 0pt
   \parsep 0pt}
   \def\newblock{}
   \sloppy
   \sfcode`\.=1000\relax}


\setlength{\baselineskip}{12.7pt}

\title{Reinforcement Learning with an Abrupt Model
Change}

\author{Wuxia Chen\\[12pt]
	Department of Industrial Engineering\\
	University of Pittsburgh\\
\and
Taposh Banerjee\\[12pt]
Department of Industrial Engineering\\
University of Pittsburgh\\
\and
Jemin George\\ [12pt]
DEVCOM Army Research Laboratory
\and
Carl Busart\\ [12pt]
DEVCOM Army Research Laboratory
}

\maketitle

\section*{ABSTRACT}
The problem of reinforcement learning is considered where the environment or the model undergoes a change. An algorithm is proposed that an agent can apply in such a problem to achieve the optimal long-time discounted reward. The algorithm is model-free and learns the optimal policy by interacting with the environment. It is shown that the proposed algorithm has strong optimality properties. The effectiveness of the algorithm is also demonstrated using simulation results. The proposed algorithm exploits a fundamental reward-detection trade-off present in these problems and uses a quickest change detection algorithm to detect the model change. 
Recommendations are provided for faster detection of model changes and for smart initialization strategies. 


\section{INTRODUCTION}
\label{sec:intro}

We study the problem of reinforcement learning (RL) with model changes in this paper. In an RL problem, an agent interacts with an environment by taking a sequence of actions to learn the optimal way to interact and optimize a long-term reward criterion \cite{Sutton:2018,Bertsekas:2012,Bertsekas:1996,Meyn:2022}. In many applications, the statistical or physical properties of the environment may change over time. It then becomes necessary for the agent to adapt its strategy to the changes.  For example, in an inventory control problem, the decision-maker has to consider the time-varying distribution of the demands to achieve the maximum possible profit. In an autonomous driving system, an autonomous car has to derive the driving policy considering the position and velocity of other vehicles \cite{Guan:2018} and also adapt to changing weather conditions. In a recommendation system, the agent must adapt its recommendations based on changing user preferences. In the framework of Markov decision processes (MDP), a change in the environment may correspond to a change in the transition probabilities of the Markov process being controlled or a change in the reward process. 

The problem of RL in a nonstationary environment has been extensively studied in the literature. We refer the readers to \cite{Banerjee:2017} and the references therein to review the literature. Some more recent references are discussed below. In the MDP context, if the transition probabilities of the model are known and the distribution of the change points (the times at which the model changes) are also known, then the problem can be reformulated in a Partially Observable MDP (POMDP) framework where a hidden state can be used to represent the true model. However, such model information is rarely known in practice. 

In this paper, we provide a model-free solution to this problem and demonstrate its performance through examples. The proposed solution is based on strong theoretical arguments. Specifically, our contributions are as follows:
\begin{enumerate}
    \item We first argue that under reasonable assumptions it is $\epsilon$-optimal to execute the optimal policy for the learned model, use a quickest change detection (QCD) algorithm \cite{Veeravalli:2014} to detect the model changes and switch to learning a new model after a model change is detected. The $\epsilon$-optimality is established by comparing the performance with that of an oracle that knows the location of change points, see Section \ref{problem_formulation}. 
    \item Next, we show that the policy that is optimal for optimizing rewards may not be optimal to detect the model change. Thus, there exists a trade-off between detection and immediate reward optimization that can be exploited to optimize the overall reward. Our proposed algorithm exploits this trade-off, see Section~\ref{sec:QCD}. 
    \item In the above context, we show that in problems like inventory control, there exists a universal policy that helps detect the model change fastest. The universal policy there corresponds to the policy that keeps the inventory full at all times, see Section~\ref{sec:univQCD}. 
    \item We also show that we can use the structural results from the MDP literature to initialize the system after a change is detected, see Section \ref{inv}. We demonstrate through simulation results that this leads to faster convergence and better overall reward, see Section~\ref{sec:smartINIT}. 
\end{enumerate}
The existence of the reward-detection trade-off was first reported in \cite{Banerjee:2017}, where a model-based solution is provided. We show in this paper that the benefit of this trade-off can be exploited even in the model-free setting by carefully designing the algorithm. In addition, in this paper, we show the existence of universal change detection policies and also discuss smart initialization strategies.

We note that it is sometimes possible for the agent to detect the model or environment change using an external sensor. For example, a change in driving conditions (e.g., weather, friction, or traffic conditions) for an autonomous cars can often be detected using external sensors. The more challenging problem is when the change in the model can only be observed through the state of the system. For example, a change in user preference or demand may not always be detected using an external sensor. In this paper, we focus on the latter problem.

While the problem has been extensively researched in the literature, its natural analytical complexity has made it challenging to solve directly. A lot of  previous works have focused on developing approximate solutions. For instance,  \cite{Hadoux:2014} and \cite{dayan1996exploration} have reformulated the problem as a Partially Observable MDP (POMDP) and utilized approximate POMDP to solve the problem.  In \cite{da2006dealing} and \cite{doya2002multiple}, they keep the estimates of the current MDP parameters and use the next state or reward to evaluate whether the parameters of the current MDP have changed. Other approaches, such as hidden mode MDPs and mixed observable MDPs (MOMDP),  have been employed  to effectively capture the transition between distinct MDPs, and obtain an  approximation  solution, see \cite{chades2012momdps} and \cite{choi2001hidden}. The changes in the properties of MDPs, such as transition kernels or rewards, will lead  to alterations in the state-action sequence. In \cite{Allamaraju:2014} and \cite{Hadoux:2014}, sequential detection methods were employed where the optimal policy for each MDP was executed, and a change detection algorithm was utilized to detect model changes, but they have not paid attention to the detection-reward trade-off. Other papers where a QCD approach is considered are \cite{dahlin2022controlling,chen2022change}. A formal and extensive comparison with other proposed solutions is part of our future work. We see our method as another candidate for an off-the-shelf algorithm with theoretical guarantees that a user can try in their RL problem.

\section{Problem Formulation and $\epsilon$-Optimal Policies}\label{problem_formulation}
Suppose we have a family of Markov Decision Processes $\mathbf{\{M_\theta\}}$, where $\theta$ takes value in some index set $\Theta$. For each $\theta$, one MDP $\mathbf{M_\theta} = (S,A,T_\theta, R_\theta )$ is defined by a tuple with four components: state space $S$, action space $A$, transition kernel $T_\theta$, and reward function $R_\theta$ \cite{Banerjee:2017}. We observe a sequence of states $\{ S_t\}$, and for each observed state $ S_t$, we make a decision $A_t$. For each state-action pair $(S_t, A_t)$, the next state $S_{t+1}$ is acquired according to the transition kernel $T_\theta$, where  \\ 
\begin{equation}\label{eq1}
T_\theta(s,a, s')= \mathbf{P}(S_{t+1}=s'|A_t=a, S_t=s). 
\end{equation}
The reward $R_\theta (S_t, A_t, S_{t+1})$ is acquired after observing the next state $S_{t+1}$. When the context is clear, we simply refer to the reward at time $t$ by $R_t$. 

In a non-stationary environment, the transition kernel and reward structure change over time. For simplicity and ease of exposition, in this paper we restrict our attention to only one change point. At some time $\gamma$, the MDP parameter changes from $\theta=\theta_0$ to $\theta=\theta_1$: 
\begin{equation}\label{eq2}
\mathbf{P}(S_{t+1}=s'|A_t=a, S_t=s) = \\ 
\begin{cases}
T_{\theta_0}(s,a, s'), \quad  t< \gamma, \quad \text{Model }\mathbf{M_0}\\
T_{\theta_1}(s,a, s'), \quad  t\geq \gamma, \quad \text{Model }\mathbf{M_1}.
\end{cases}
\end{equation}
A policy is defined as the potentially infinite-length vector of Markov maps:
$$ 
\Pi = [\mu_0, \mu_1, \dots],
$$ 
where each $\mu_t$ is a map from state $S_t$ to action $A_t$. 

If the model is stationary and the parameter $\theta$ remains the same, then one of the classical ways to solve the MDP problem is to seek a policy to maximize the long term discounted reward:
\begin{equation}\label{eq:stationary}
	J_\theta^*(s_0) = \max_\Pi \;  \mathbf{E}_\theta\left[\sum_{t=0}^\infty \beta^t R_{t} \; \Big| \; S_0=s_0\right].
\end{equation}
where $\beta \in (0,1)$ is a discount factor and the expectation is with respect to the true $\theta$. We will use $\Pi_\theta^*=[\mu_\theta^*, 
\mu_\theta^*, \dots]$ to denote the optimal stationary policy for this problem when the true model parameter is $\theta$. 

In a non-stationary environment where the MDP changes from $\mathbf{M_0}$ to  $\mathbf{M_1}$ ($\theta_0$ to $\theta_1$) at change point $\gamma$ ($\gamma$ is unknown to the agent), we modify the discounted cost problem as
\begin{equation}\label{eq:nonstationary}
	J_{\theta_0, \theta_1}^*(s_0) = \max_\Pi \;  \mathbf{E}_{\theta_0, \theta_1}\left[\sum_{t=0}^{\gamma-1}\beta^t R_{t} + \sum_{t=\gamma}^{\infty} \beta^ {t-\gamma}R_{t} \; \Big| \; S_0=s_0\right].
\end{equation}
This way of resetting the discounting gives equal weights to the performance of a policy before and after the change. 

We now define the concept of an oracle:
\begin{definition}[Oracle Policy]
A policy is called an oracle policy if it has knowledge of the change point $\gamma$, and executes the policy $\Pi_{\theta_0}^*$ before change and the policy $\Pi_{\theta_1}^*$ after the change. 
\end{definition}
It is clear that if the change point $\gamma$ is large enough, the discounted reward for an oracle policy is approximately equal to the value $J_{\theta_0, \theta_1}^*(s_0)$ for any initial state $s_0$. We now show that using a quickest change detection algorithm \cite{Veeravalli:2014} one can achieve the performance of an oracle under modest assumptions on the problem. 

\begin{definition}[Quickest Change Detection (QCD) Algorithm and QCD-based policy]
    By a QCD algorithm we mean a stopping time $\tau$ whose value is decided based on the sequence of states $S_0, S_1, \dots$, actions $A_0, A_1, \dots$, and rewards $R_0, R_1, \dots$. At time $\tau$ we declare that a change in model has occurred. A policy that employs a QCD stopping rule to detect change is called a QCD-based policy. 
\end{definition}

We define the information number \cite{Banerjee:2017,lai1998information}
\begin{equation}
\label{eq:informationnumber}
    I_{\theta_0, \theta_1} = \lim_{n \to \infty} \; \frac{1}{n} \sum_{k=1}^n \frac{T_{\theta_1}(S_k, A_k, S_{k+1})}{T_{\theta_1}(S_k, A_k, S_{k+1})}, \quad \text{when } \gamma =1.
\end{equation}

\begin{theorem}[$\epsilon$-optimality of QCD-based policy]
\label{thm:epsilonopt}
When the transitions functions before and after change are known, then to achieve $\epsilon$ optimality in the problem in \eqref{eq:nonstationary}, it is enough to restrict our search to QCD-based policies. 
\end{theorem}
\begin{proof} We only provide a sketch of the proof. We construct a policy and show that under certain conditions it is $\epsilon$-optimal for the problem in \eqref{eq:nonstationary}. Consider the policy that initially employs $\Pi_{\theta_0}^*$ before change, detects the change using a QCD stopping rule, and then switches to the policy $\Pi_{\theta_1}^*$ after the change is detected. If the rewards are bounded by $M$ and the optimal QCD algorithm is used, then the performance of this algorithm will be within $\mathbf{E}(\tau-\gamma)M + \delta$ of the oracle. Here $\delta$ is a small positive constant. It is well-known that the detection delay of the optimal algorithm is inversely proportional to the information number $I_{\theta_0, \theta_1}$. Thus, if this number if large enough, the term  $\mathbf{E}(\tau-\gamma)M$ will be small enough. Since the oracle policy is $\epsilon$-optimal, so is the proposed policy. 
\end{proof}

In practice, we do not know the models and hence cannot directly use the optimal policy and also cannot employ the optimal QCD algorithm. However, we can use algorithms like Q-learning to learn the optimal policy and use nonparametric methods in QCD to achieve the oracle performance. In the rest of the paper, we make the assumption that the change point $\gamma$ is large enough so that the Q-learning algorithm has a reasonable amount of time to converge to the optimal policy. In other words, we assume that the stationarity is slowly changing. 



\section{Q-learning algorithm with decreasing epsilon greedy action selection}\label{Q-learning}
We use a modified version of the classical Q-learning algorithm \cite{Watkins:1992,Bertsekas:1996} to learn the optimal policy for each model, before and after the change. 
We provide a brief overview of the modified Q-learning algorithm here. 

It is well-known that the optimal reward function $J_\theta^*$ satisfies the Bellman equation. 

\begin{equation}\label{eq:bellman}
	J_\theta^*(s) = \max_{a} \sum_{s'} T_\theta(s,a,s') \Big[R_\theta(s, a, s') + \beta J_\theta^*(s')\Big].
\end{equation}
The $Q$-function is defined as 
\begin{equation}\label{eq:Qfunc}
    Q_\theta^*(s,a) = \sum_{s'} T_\theta(s,a,s') \Big[R_\theta(s, a, s') + \beta J_\theta^*(s')\Big].
\end{equation}
With the definition, the $Q$-function also satisfies a fixed-point equation given by
\begin{equation}\label{eq:Qfixedpt}
    Q_\theta^*(s,a) = \sum_{s'} T_\theta(s,a,s') \Big[R_\theta(s, a, s') + \beta \max_{a'} Q_\theta^*(s', a')\Big].
\end{equation}
The problem of Q-learning is to estimate, for any fixed $\theta$, the optimal Q-function $Q_\theta^*(s,a)$ without knowing the transition function $T_\theta(s,a,s')$. In the $Q$-learning algorithm this estimation is done using a stochastic approximation algorithm \cite{borkar2009stochastic,harold1997stochastic}. 

For a sequence of states and actions $S_0, A_0, S_1, A_1, S_2, A_2, \dots$, the Q-learning algorithm estimates the $Q$-function for each state-action pair using the updates 
\begin{equation}
TD \leftarrow R_t+ \beta \max_{ a } Q(S_{t+1},a )-   Q(S_{t},A_t )  
\end{equation}
\begin{equation}
Q(S_{t},A_t ) \leftarrow Q(S_{t},A_t ) + \alpha TD,
\end{equation}
where  $Q(S_{t},A_t )  $ is the current Q value of state-action pair $(s,a)$, $ \max_{ a } Q(S_{t+1},a ) $ is the estimation of optimal future value, $R_t$ is the received reward when taking action $A_t$ at state $S_t$,  $\beta$ is the discount factor $(0\leq \beta \leq 1)$, and $\alpha$ is the learning rate $(0< \alpha \leq 1)$. For the $Q$-learning to converge to the optimal $Q^*$, we must visit each state-action pair infinitely often \cite{Bertsekas:1996}. This is achieved by using an $\epsilon$-greedy strategy where a random action is chosen with probability $\epsilon$ and the optimal action (based on current estimates of $Q$) is chosen with probability $1-\epsilon$. To ensure faster convergence in a nonstationary environment, we use a variant of $Q$-learning in which the learning rate and the exploration rate are reduced over time. The entire algorithm is given in Algorithm 1.

\begin{algorithm}
    \SetKwInOut{KwIn}{Data}
    \SetKwInOut{KwOut}{Result}
    \KwIn{initialized Q-table, initial learning rate $\alpha_0$, initial exploration rate $\epsilon$, discount factor $\beta$, cut-off greedy probability $\epsilon_c$, cut-off learning rate $\alpha_c$.}
    \KwOut{a trained Q-table that gives optimal (or near-optimal) policy }
    initialization: set $S_0 = 0$, learning rate $\alpha=\alpha_0$, $\epsilon= \epsilon_0 $   \\
    
    \For{k =  1 : time-horizon}{
        $c \leftarrow{random(0,1)} $  \\
        \eIf{$c < \epsilon$}{
            $a \leftarrow random.choice(\text{action space A})$
         }{
            $a \leftarrow \arg \max_{a}(Q(s,\cdot))$
         }    
        $ TD \leftarrow R_k+ \beta \max_{ a' } Q(s',a' )-   Q(s,a ) $ \\ 
        $ Q(s,a) \leftarrow Q(s,a ) + \alpha TD $ \\
        $ s \leftarrow{s'} $ \\
        \If{$\epsilon > \epsilon_{c}$} {
             $ \epsilon \leftarrow{ \epsilon- \Delta } $ \quad \quad \quad \quad \tcc{   decrease exploration rate} 
             }   
        \If{$\alpha > \alpha_{c}$} {
             $ \alpha \leftarrow{ \alpha- \Delta } $ \quad \quad \quad \quad \tcc{decrease learning rate } 
             } 
    }
    \KwRet{a new Q-table whose $ \arg \max_{a}(Q(s, \cdot) $  gives an optimal or near optimal policy}
    \caption{Q-learning algorithm with decreasing epsilon greedy action selection}
\end{algorithm}

\section{Quickest Change Detection Algorithms and Effect of Policy}
\label{sec:QCD}
To detect a model change using the state and reward processes, we use a quickest change detection algorithm \cite{Banerjee:2017,Veeravalli:2014,tartakovsky2014sequential}.
If the transition kernels $T_{\theta_0}(s,a,s')$ and $T_{\theta_1}(s,a,s') $ are known and it is also known that the model will change from parameter $\theta_0$ to $\theta_1$, then we can use the generalized cumulative sum (CUSUM) algorithm from \cite{lai1998information}. In this algorithm, we compute a sequence of statistics
$$
W_n = \max_{1 \leq k \leq n} \sum_{i=k}^n \frac{T_{\theta_1}(S_{i-1},A_{i-1},S_{i}) }{T_{\theta_0}(S_{i-1},A_{i-1},S_{i}) },
$$
and stop the first time this statistic is above a pre-defined threshold:
$$
\tau_c = \min\{n \geq 1: W_n > A\}. 
$$
If $\gamma$ is a constraint on the mean time to a false alarm, then it has been shown in \cite{lai1998information} that under mild conditions, the delay of the generalized CUSUM algorithm is given by
\begin{equation}
    \label{eq:CUSUMperf}
    \mathbf{E}_1[\tau_c] = \frac{\log \gamma}{I_{\theta_0, \theta_1}}, \quad \textbf{as } \gamma \to \infty,
\end{equation}
where $I_{\theta_0, \theta_1}$ is defined in \eqref{eq:informationnumber}. 
Note that the delay is inversely proportional to the information number $I_{\theta_0, \theta_1}$. This number depends on the policy through the choice of action sequence $\{A_t\}$. Thus, different policies will lead to different values of $I_{\theta_0, \theta_1}$ and hence different detection delays. In general, this characteristic and dependence on information number is shown by almost all popular QCD algorithms. 

If the state and action spaces are finite, there must exist an optimal policy for quickest change detection. We note that the best policy may depend on the algorithm used for QCD. 
\begin{definition}[Best QCD Policy]
    A policy is called the best QCD policy if when applied to the system leads to the fastest detection of a model change. 
\end{definition}

Since we do not have access to the transition kernels, we cannot use the above CUSUM algorithm. If the state space is high-dimensional (or even moderate-dimensional), then tracking changes in the state-space model becomes intractable. As a result, we use the nonparametric CUSUM algorithm \cite{basseville1993detection} applied to the reward process $\{R_k\}$. The stopping rule remains the same, but we computer the statistic $W_n$ using (for example) 
$$
W_n = \max\big\{0, W_{n-1} + R_n - \mu_0 - \eta \sigma_0\big\}.
$$
The algorithm works as follows. The parameters $\mu_0, \sigma_0$ are the (estimated) mean and standard deviation of $R_n$ before the change, and $\eta$ is a control parameter. Before the change, $R_n$ and $\mu_0$ cancel each other giving the reflected random walk $W_n$ a negative drift. After the change, if the average reward increases more than  $\eta \sigma_0$, then the drift becomes positive and can be detected using a large positive threshold $A$. Thus, $\eta$ controls the amount of change in the average reward that we would like to tolerate before sounding an alarm. We note that the notion of Best QCD policy is well-defined even when we use the nonparametric CUSUM algorithm: it is the policy that leads to the fastest delay when using the algorithm. 

\section{$\epsilon$-Optimal Policies and Exploiting Reward-Detection Trade-off}\label{change_detection}
Based on the result in Theorem~\ref{thm:epsilonopt} on the $\epsilon$-optimality of QCD-based policies, we can argue that it is enough to restrict our search to this class of policies. In addition, the discussion in the previous section suggests that the reward-detection trade-off should be exploited and the Best QCD policy should be used to achieve better performance. In \cite{Banerjee:2017} it was shown that in the model-based setting, this exploitation is possible and leads to better rewards. It is not clear \textit{a priori} that this trade-off can be used even when the model parameters are not known. In addition to the fact that using the Best QCD policy is not optimal for rewards, we also learn the best policy locally using $Q$-learning.

We show in this paper that exploitation is possible even in the model-free or RL setting. To demonstrate this, we compare two basic algorithms, one in which the Best QCD policy is used, and another, in which it is not used. 



\subsection{Single-Threshold Change Detection: A Policy without using Best QCD Policy}
In this section, we propose an end-to-end algorithm for RL with model changes. We call the algorithm the Single-Threshold Adaptive $Q$-Learning (STAQL) algorithm. 
In STAQL, we first learn the optimal policy for $\mathbf{M}_0$ using our $Q$-learning algorithm discussed in Algorithm 1. We initialize the $Q$-matrix using numbers that can help with achieving faster convergence and more rewards. In the next section, we discuss how to smartly initialize an inventory control system. The system starts at time $0$. Let $\tau$ be the time at which the $Q$-learning converges and learns the optimal policy for model $\mathbf{M}_0$. We can learn the time $\tau$ through simulations and experience with the system. From time $\tau$ to another time $\delta$ we learn the baseline reward statistics and estimate
$$
\mu_0= mean(R[\tau:\delta]), \quad \quad \sigma_0=sd(R[\tau:\delta]).
$$
Here $R[\tau:\delta]$ denotes the vector of rewards collected from time $\tau$ to $\delta$. 
Starting time $\delta$, we apply the nonparametric CUSUM algorithm to detect the model change. Here, we are assuming that $\gamma$, the change point, satisfies $\gamma \gg \delta$. If the average reward is expected to change from low to high, we use the statistical update
\begin{equation}\label{cdhl}
W_n = \max\big\{0, W_{n-1} + R_n - \mu_0 - \eta \sigma_0\big\}.
\end{equation}
or if the average reward is expected to change from high to low, we use the following instead,
\begin{equation}\label{cdlh}
W_n = \min\big\{0, W_{n-1} + R_n - \mu_0 + \eta \sigma_0\big\}.
\end{equation}
If the average reward can change in any direction, we can use both statistics in parallel. The change is declared at
$$
\hat{\gamma} = \min\{n \geq \delta: W_n > A\}. 
$$
After the change is detected at $\hat{\gamma}$, we reinitialize the $Q$-matrix to smart values and start the $Q$-learning again (Algorithm 1) to learn the optimal policy for model $\mathbf{M}_1$. The STAQL algorithm is written in algorithmic form in Algorithm 2 and can also be represented using the following equation:
\begin{equation}
\Pi_{\text{STAQL}}=( \underbrace{\tilde{\pi}, ...,\hat{\pi}_0, \hat{\pi}_0, ...}_{ |w|\leq A, \hat{\gamma}-1},   
 \underbrace{\tilde{\pi}, ..., \hat{\pi}_1,\hat{\pi}_1,...}_{|w|> A, \hat{\gamma} \text{ onward} } ),
\end{equation}
where $\tilde{\pi}$ is the Markov map generated from the initial $Q$-table, $\hat{\pi}_0$ (respectively, $\hat{\pi}_1$) is the optimal policy for model $\mathbf{M}_0$ (respectively, $\mathbf{M}_1$) learned using $Q$-learning. 


\begin{algorithm}
    \SetKwInOut{KwIn}{Presets}
    \SetKwInOut{KwOut}{Result}
    \KwIn{threshold A, and stabilizer $\eta $}
    \KwOut{a detected change point $\hat{\gamma}$, discounted reward }
    initialization $S_0 =0, w=0$\\
    found=False \\
    set a smart initial Q-table according to the demand \\  
    \For{t =  1 : time-horizon}{
    do one step Q-learning, updating the Q-table according to the transition kernel\\
    document each step reward $R_t$ \\
    \If {$t==\delta$}   {
        take the single step reward $R[\tau:\delta] $ as the bench mark, and compute the mean and standard deviation $\mu_0= mean(R[\tau:\delta])$, $\sigma_0=sd(R[\tau:\delta])$ \\
    }  
    \If{$ t> \delta $ \& not found}  {
    compute the change detector $w$ \\
    $ w\leftarrow \max (0, w+R_t-\mu_0 - \eta sd_0)$ (if average reward changes from low to high) \\
    or $  w\leftarrow \min (0, w+R_t-\mu_0 + \eta sd_0) $ (if average reward changes from high to low)\\
    \If{$|w|>A$} {
    found=True\\   \tcc{change is detected!} 
    document the detected change $\hat{\gamma}\leftarrow t $\\
    reset Q-table according to the next stage demand \\
    reset learning parameters $\epsilon$, $\alpha$\\
                     } 
        } 
    } 
     \KwRet{ detected change time $\hat{\gamma}$, discounted reward }
    \caption{Single-Threshold Adaptive $Q$-Learning (STAQL) algorithm}
\end{algorithm}


\subsection{Two-Threshold Change Detection: Policy Exploiting Reward-Detection Trade-Off}
In this section, we propose another end-to-end algorithm for RL with model changes. We call the algorithm the Two-Threshold Adaptive $Q$-Learning (TTAQL) algorithm. The TTAQL exploits the reward-detection trade-off and uses the Best QCD policy to optimize the overall reward.

\begin{algorithm}[H]
    \SetKwInOut{KwIn}{Presets}
    \SetKwInOut{KwOut}{Result}
    \KwIn{quick change detection policy $\pi_{qcd}$, threshold B,  threshold $\tilde{A}$ and stabilizer $\eta $}
    \KwOut{a detected change point $\hat{\gamma}$, discounted reward }
    initialization $S_0 =0, w=0$\\
    found=False \\
    suspect=False \\
    set a smart initial Q-table according to the demand \\  
    \For{ t =  1 : time-horizon}{
    \eIf{not suspect} {
    do one step Q-learning, updating the Q-table according to the transition kernel\\}
    {
        apply $\pi_{qcd}$ policy \\
        update the state $S_t$ according to the transition kernel\\
         \tcc{do not update the Q-table while applying QCD policy} 
    }
    document each step reward $R_t$ \\
    
    \If {$t==\delta$}   {
        take the single step reward $R[\tau:\delta] $ as the bench mark, and compute the mean and standard deviation $\mu_0= mean(R[\tau:\delta])$, $\sigma_0=sd(R[\tau:\delta])$ \\
    }  
    
    \If{$ t> \delta $ \& not found}  {
    compute the change detector $w$, using Eq (\ref{cdhl}) or Eq (\ref{cdlh})\\
    \eIf{$|w|>B $ }{
    suspect=True
    }{suspect=False}
    
    \If{suspect} {
    \If{$|w|>\tilde{A} $} {
        found=True\\   \tcc{change is detected!} 
        suspect=False \\
        document the detected change $\hat{\gamma}\leftarrow t $\\
        reset Q-table, and learning parameters $\epsilon$, $\alpha$\\
    }
                     } 
        } 
    } 
     \KwRet{ detected change time $\hat{\gamma}$, discounted reward }
    \caption{Two-Threshold Adaptive $Q$-Learning (TTAQL) algorithm}
\end{algorithm}


Similar to STAQL, the TTAQL algorithm also uses a smart initialization followed by $Q$-learning to learn the optimal policy for $\mathbf{M}_0$. 
Again, similar to the STAQL algorithm, the TTAQL algorithm learns the baseline reward statistics using
$$
\mu_0= mean(R[\tau:\delta]), \quad \quad \sigma_0=sd(R[\tau:\delta]).
$$

Starting time $\delta$, however, a two-threshold version of the nonparametric CUSUM algorithm is applied to detect the model change. To clarify concepts, we assume that the average reward is expected to decrease. We will then use the statistic
\begin{equation}\label{cdlh2}
W_n = \min\big\{0, W_{n-1} + R_n - \mu_0 + \eta \sigma_0\big\}.
\end{equation}
The change is declared at
$$
\hat{\gamma} = \min\{n \geq \delta: W_n > A\}. 
$$
However, the algorithm uses another threshold $B < A$ to choose which policy to use at any time after time $\delta$. Specifically, if $\hat{\pi}_0$ denotes the Markov map for model $\mathbf{M}_0$ learned using $Q$-learning and $\pi_{qcd}$ denotes the Markov map for the Best QCD policy, then we use the following strategy:
\begin{equation}
\begin{split}
    \text{If }&\quad 0 \leq |W_n| \leq B, \quad \text{use map } \hat{\pi}_0 \text{ at time } n+1\\
    \text{If }&\quad B < |W_n| < A, \quad \text{use map } \pi_{qcd} \text{ at time } n+1.
    \end{split}
\end{equation}
After the change is detected at $\hat{\gamma}$, we reinitialize the $Q$-matrix to smart values and start the $Q$-learning again (Algorithm 1) to learn the optimal policy for model $\mathbf{M}_1$. The TTAQL algorithm is written in algorithm form in Algorithm 3, can also be represented using the following equation:
\begin{equation}
\Pi_{\text{TTAQL}}=(\underbrace{\tilde{\pi},...,\hat{\pi}_0,}_{|w|<B} \; \underbrace{\pi_{qcd},...\pi_{qcd}}_{B<|w|<\tilde{A}}, \; \underbrace{\hat{\pi}_0, \hat{\pi}_0,...,}_{|w|<B} 
\underbrace{\pi_{qcd},...\pi_{qcd},}_{B<|w|<\tilde{A}, \text{ before } \hat{\gamma}-1} \;
\underbrace{\tilde{\pi}, ..., \hat{\pi}_1,\hat{\pi}_1,...}_{|w|>\tilde{A}, \; \hat{\gamma} \text{  onward} }  ).
\end{equation}

\section{Simulation Results: Application to an Inventory Control Problem}\label{inv}
In this section, we apply the STAQL and TTAQL algorithms to an inventory control problem and show that the TTAQL algorithm can outperform the STAQL algorithm. We also show that the Best QCD policy for this problem is universal: there exists an interpretable policy that can detect the change fastest for any realization of the inventory control problem. We also discuss the convergence rate for smart initializations. 

\subsection{Inventory control problem}
Consider the inventory control problem with inventory level or state $S_t$, $S_t \in \{0,1,...N\}$, and $N$ is the maximum inventory size of the warehouse. Let action $A_t$ be the number of new orders in the morning of day $t$, $A_t \in \{0,1,...N\}$. During the day, customers come with a stochastic demand $D_t$, where $D_t$ is an independent and identically distributed sequence of Poisson random variables with some rate $\lambda$: 
\begin{equation}
 S_{t+1}=   max (min(S_t+A_t, N)-D_t, 0).
\end{equation}
The reward or income on the day $t$ is 
\begin{equation}
    R_{t}= -k \vmathbb{I}(A_t>0) - c(min(A_t, N-S_t)) -hS_{t+1}+p(min(S_t+A_t, N)-S_{t+1})-\text{rent}. 
\end{equation}
The income on the day $t$ is determined as follows: there is a fixed entry cost $k$ of ordering nonzero items and each item must be purchased at a fixed price $c$, so the cost associated with purchasing $A_t$ items is $k \vmathbb{I}(A_t>0) + cA_t$. In addition, there is a cost of $h$ for holding an unsold item. If there are $x$ leftovers at the end of the day $t$, the manager will pay $hx$ for holding the items the next morning.  Finally, upon selling $z$ items the manager receives a payment of $p z$. To make the warehouse running, we must have $p > h$, otherwise, there is no incentive to order new items.

We use the following set of parameters for the simulations. The parameters we use here are: warehouse capacity $N=5$, discounted factor $\beta=0.9999$, change point $\gamma=1000$, $time horizon =5000$, stabilizer $\eta=0.92$, $A=6 sd_0$, $B=3.35 sd_0$, $\tilde{A}=6.67 sd_0$, fixed cost $k=0.5$, per unit cost $c=3$, holding cost $h=2$, profit per sale $p=8$, and rent=4.8, initial learning rate $\alpha_0=0.2$, cut-off learning rate $\alpha_c=0.05$, initial exploration rate $\epsilon_0=0.2$, cut-off exploration rate $\epsilon_c=0.05$, step decent $\Delta=0.001$. We use the $R_t, \forall t\in [500,600]$ as the bench mark, where $\mu_0= mean(R[500:600])$, $sd_0 = sd(R[500:600])$. All results are averaged over 10,000 iterations.




\subsection{Using Smart Initialization in $Q$-Learning}\label{sec:smartINIT}
It is well known that the optimal Markov map in the inventory control problem is linear and nonincreasing \cite{Bertsekas:2012}. 
In Fig.~\ref{opt_initial_high}, we show that if we initialize the $Q$-tables with values that correspond to a monotonically decreasing policy that it leads to faster convergence (left figure) and better overall rewards (right figure). In the figure, $Q$-random corresponds to a randomly initialized $Q$-table, and $Q$-pyramid corresponds to a policy that is unimodal, increasing first, and then decreasing after the mode. We see this pattern as long as the demand is high. If the demand is low, we have observed that initializing with $Q$-random leads to the best overall reward.

\begin{figure}[htbp]
\begin{center}
\includegraphics[width=0.43\textwidth]{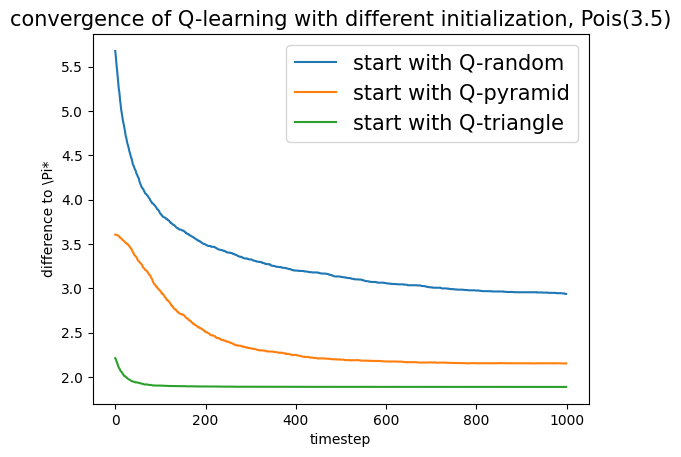}
\includegraphics[width=0.4\textwidth]{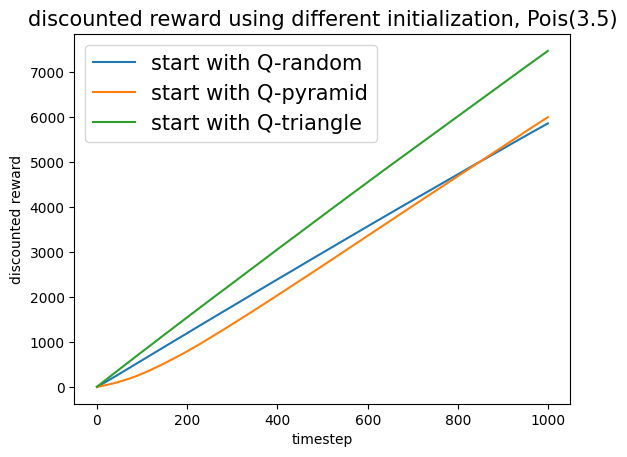}
\caption{Q-learning with different initialization.}
\label{opt_initial_high}
\end{center}
\end{figure}

\subsection{A Universal Change Detection Policy}\label{sec:univQCD}
In general, for every realization of the problem, one may have to search for the Best QCD policy through simulations. In the inventory control problem, however, we show that there is a universal Best QCD policy that is Best QCD policy for every realization of the inventory problem. This policy corresponds to the one that keeps the inventory full at all times. This is intuitive since if the demand is low, it is optimal to keep the inventory low. If there is a sudden increase in demand, items may appear out of stock and a user may never place an order. The system will fail to detect a sudden increase in demand from low to high. However, if we keep the inventory full at all times, we can always capture the fluctuations in demand.

Table~\ref{QCD_HL} compares the detection performances of the Best QCD policy and the learned optimal policy in the situation where the demand for $\textbf{M}_0$ is high and then it switches to a lower demand after the change point $\gamma=1000$. 
We note that in the table, the learned optimal policy $\hat{\pi}_0$ changes with the choice of demand rates. In Table~\ref{QCD_LH}, we show similar results when the demand of $\textbf{M}_0$ is low and the demand for $\textbf{M}_1$ is high. 

\begin{table}[htbp]
\caption{Best QCD policy vs $\hat{\pi}_0$: demand from high to low}
\begin{center}
\begin{tabular}{|c c c c c c |} 
 \hline
high to low  & $\eta$ & Best QCD policy delay & FA & $\hat{\pi}_0$  delay & FA   \\ [0.5ex] 
 \hline
$\lambda_0$=4,  $\lambda_1$=1.8 & 0.92 & 96  & 0.009 & 228 &0.009\\ 
 \hline
& 0.7 & 26  & 0.008 & 53 &0.008\\ 
 \hline
$\lambda_0$=3, $\lambda_1$=1  & 0.9 & 48.9  & 0.0091 & 165 & 0.0096\\ 
 \hline
& 0.7 & 17.5  & 0.0089 & 22 & 0.0094\\ 
 \hline
$\lambda_0$=3.5, $\lambda_1$=2.5  & 0.2 & 89  & 0.0069 & 170 & 0.0073\\ 
 \hline
& 0.1 & 109  & 0.0084 & 160 & 0.0101\\ 
 \hline
\end{tabular}
\end{center}
\label{QCD_HL}
\end{table}

\begin{table}
\caption{Best QCD policy vs $\hat{\pi}_0$: demand from low to high}
\begin{center}
\begin{tabular}{|c c c c c c |} 
 \hline
low to high  & $\eta$ & Best QCD policy delay & FA & $\hat{\pi}_0$  delay & FA   \\ [0.5ex] 
 \hline
$\lambda_0$=2,$\lambda_1$=4 & 0.3 & 19  & 0.0066 & 175 & 0.0101\\ 
 \hline
& 0.1 & 20  & 0.0088 & 100 & 0.0092\\ 
 \hline
$\lambda_0$=1.5,$\lambda_1$=3.5  & 0.4 & 13  & 0.0074 & 95 & 0.0083\\ 
 \hline
& 0.3 & 12  & 0.0084 & 35 & 0.0096\\ 
 \hline
$\lambda_0$=2,$\lambda_1$=3  & 0.2 & 63  & 0.0074 & 344 & 0.0092\\ 
 \hline
& 0.05 & 50  & 0.01 & 225 & 0.01\\ 
 \hline
\end{tabular}
\end{center}
\label{QCD_LH}
\end{table}

\subsection{Comparison of STAQL and TTAQL Policies}\label{select_w}
In the previous section, we showed that the Best QCD policy can detect changes faster. However, it can also cause a loss of immediate rewards. In Fig.~\ref{N5} we show that the TTAQL algorithm can also achieve a better overall reward. In the figure, oracle policy refers to the oracle policy discussed in Definition 1 except the policies are learned using $Q$-learning. The ignore policy simply ignores the change and incurs heavy losses due to a high holding cost. The table in the figure shows the expected discounted reward at the end of the horizon with Rwd(mdp1) as the reward collected beginning at the change point. 
In the figure on the right in Fig.~\ref{N5}, we 
plot the cumulative reward beginning at the change point, averaging over $10000$ realizations. The figure shows that the TTAQL performs better almost at each point in the entire horizon. To guarantee a fair evaluation of the two aforementioned methods, we maintain a false alarm rate of roughly 1\%. When computing the average delay and average reward, we exclude only the false alarm instances. Additionally, if there are any cases where the agent fails to detect a change, we assume that the agent detects the change at the last possible moment.



\begin{figure}[ht]
    \centering
    \begin{minipage}{0.45\linewidth}
        \centering
        \begin{tabular}{|c c c c c |} 
         \hline
          & TTAQL & STAQL & Ignore & Oracle  \\ [0.5ex] 
         \hline
         Rwd(mdp1) & 264 & 185 & -3210 & 424     \\ 
         \hline
        Rwd(total) & 8376 & 8310 &  4928  &  8601  \\  
         \hline
        avg-delay  & 145 & 227 & $\infty$ &  0 \\ 
         \hline
        true-detect\%  & 97.78 & 96.78 & 0 &  1 \\ 
         \hline
        miss\%  & 1.16 & 2.09 &  & \\ 
         \hline
        F-A \% & 1.06 & 1.13 &  &  \\ 
         \hline
        \end{tabular}
    \end{minipage}
    \hspace{0.05\linewidth}
    \begin{minipage}{0.45\linewidth}
        \centering
        \includegraphics[width=\linewidth]{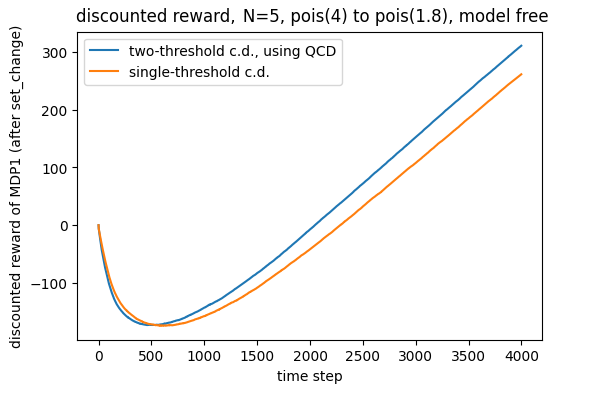}
    \end{minipage}
    \caption{Discounted reward and delay using different change detection policies, $\lambda_0=4, \lambda_1=1.8$ N=5}
    \label{N5}
\end{figure}

In Fig.~\ref{N7}, we show the results for the inventory control problem with the maximum warehouse capacity $N=7$. Most of the parameters are the same to the $N=5$ case, except for $\lambda_0 =6$, $\lambda_1 =2.5$, $\eta=1.2$, $B=4 sd_0$, $A=8 sd_0$, and $\tilde{A}=6.9 sd_0$.



\begin{figure}[ht]
    \centering
    \begin{minipage}{0.45\linewidth}
        \centering
        \begin{tabular}{|c c c c c |} 
         \hline
          & TTAQL & STAQL & Ignore & Oracle  \\ [0.5ex] 
         \hline
         Rwd(mdp1) & 603 & 521 & -3838 & 718     \\ 
         \hline
        Rwd(total) & 10645 & 10552 &  6251  &  10812  \\  
         \hline
        avg-delay  & 63 & 139 & $\infty$ &  0 \\ 
         \hline
        true-detect\%  & 98.77 & 97.78 & 0 &  1 \\ 
         \hline
        miss\%  & 0.41 & 1.26 &  & \\ 
         \hline
        F-A \% & 0.82 & 0.96 &  &  \\ 
         \hline
        \end{tabular}
    \end{minipage}
    \hspace{0.05\linewidth}
    \begin{minipage}{0.45\linewidth}
        \centering
        \includegraphics[width=\linewidth]{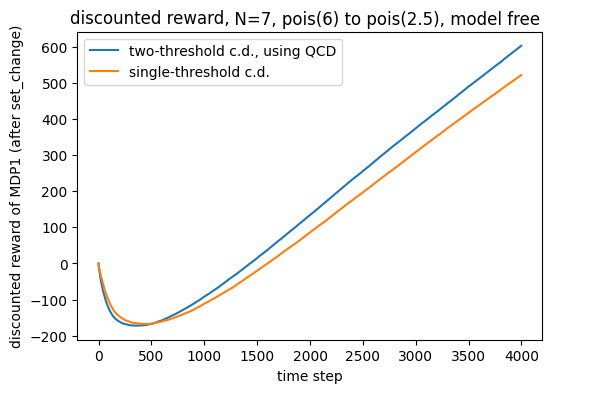}
    \end{minipage}
    \caption{Discounted reward and delay using different change detection policies, $\lambda_0=6, \lambda_1=2.5$ N=7}
    \label{N7}
\end{figure}

\section{Conclusion}
We proposed an algorithm called the Two-Threshold Adaptive $Q$-Learning (TTAQL) algorithm that can be used for RL with model changes. This algorithm exploits a fundamental trade-off between detection delay and immediate reward optimization that is present in RL in nonstationary environments. We also showed that this algorithm belongs to a class of policies called QCD-based policies. We argued in Theorem 1 that in the search for optimal policies, it is enough to restrict the search to QCD-based policies because one can achieve $\epsilon$-optimality. We also showed that in some applications like inventory control, there is a universal policy that provides the fastest detection delay. This policy can be used to exploit the reward-detection trade-off. In addition, smart initialization in the $Q$-learning algorithm can lead to faster convergence and better overall rewards. In the future, we plan to apply the TTAQL algorithm to more complex RL problems, develop better change detection algorithms for this domain, and also develop deeper theoretical insights. 



\vspace{-0.2cm}
\section{Acknowledgement}
The work of Wuxia Chen and Taposh Banerjee was supported in part by
the U.S. Army Research Lab under grant W911NF2120295.
\vspace{-0.2cm}
\footnotesize

\bibliographystyle{wsc}

\bibliography{demobib}

\begin{thebibliography}{}

\bibitem[\protect\citeauthoryear{Allamaraju, Kingravi, Axelrod, Chowdhary,
  Grande, How, Crick, and Sheng}{Allamaraju et~al.}{2014}]{Allamaraju:2014}
Allamaraju, R., H.~Kingravi, A.~Axelrod, G.~Chowdhary, R.~Grande, J.~P. How,
  C.~Crick, and W.~Sheng. 2014.
\newblock ``Human aware UAS path planning in urban environments using
  nonstationary MDPs''.
\newblock In {\em 2014 IEEE International Conference on Robotics and Automation
  (ICRA)},  1161--1167.
\newblock IEEE.

\bibitem[\protect\citeauthoryear{Banerjee, Liu, and How}{Banerjee
  et~al.}{2017}]{Banerjee:2017}
Banerjee, T., M.~Liu, and J.~P. How. 2017.
\newblock ``Quickest change detection approach to optimal control in Markov
  decision processes with model changes''.
\newblock In {\em 2017 American control conference (ACC)},  399--405.
\newblock IEEE.

\bibitem[\protect\citeauthoryear{Basseville, Nikiforov, et~al.}{Basseville
  et~al.}{1993}]{basseville1993detection}
Basseville, M., I.~V. Nikiforov et~al. 1993.
\newblock {\em Detection of abrupt changes: theory and application}, Volume
  104.
\newblock prentice Hall Englewood Cliffs.


\bibitem[\protect\citeauthoryear{Bertsekas}{Bertsekas}{2012}]{Bertsekas:2012}
Bertsekas, D. 2012.
\newblock {\em Dynamic programming and optimal control: Volume I}, Volume~1.
\newblock Athena scientific.


\bibitem[\protect\citeauthoryear{Bertsekas and Tsitsiklis}{Bertsekas and
  Tsitsiklis}{1996}]{Bertsekas:1996}
Bertsekas, D., and J.~N. Tsitsiklis. 1996.
\newblock {\em Neuro-dynamic programming}.
\newblock Athena Scientific.


\bibitem[\protect\citeauthoryear{Borkar}{Borkar}{2009}]{borkar2009stochastic}
Borkar, V.~S. 2009.
\newblock {\em Stochastic approximation: a dynamical systems viewpoint},
  Volume~48.
\newblock Springer.


\bibitem[\protect\citeauthoryear{Chades, Carwardine, Martin, Nicol, Sabbadin,
  and Buffet}{Chades et~al.}{2012}]{chades2012momdps}
Chades, I., J.~Carwardine, T.~Martin, S.~Nicol, R.~Sabbadin, and O.~Buffet.
  2012.
\newblock ``MOMDPs: a solution for modelling adaptive management problems''.
\newblock In {\em Proceedings of the AAAI Conference on Artificial
  Intelligence}, Volume~26,  267--273.

\bibitem[\protect\citeauthoryear{Chen, Tang, and Gupta}{Chen
  et~al.}{2022}]{chen2022change}
Chen, H., J.~Tang, and A.~Gupta. 2022.
\newblock ``Change Detection of Markov Kernels with Unknown Pre and Post Change
  Kernel''.
\newblock In {\em 2022 IEEE 61st Conference on Decision and Control (CDC)},
  4814--4820.
\newblock IEEE.

\bibitem[\protect\citeauthoryear{Choi, Yeung, and Zhang}{Choi
  et~al.}{2001}]{choi2001hidden}
Choi, S.~P., D.-Y. Yeung, and N.~L. Zhang. 2001.
\newblock ``Hidden-mode markov decision processes for nonstationary sequential
  decision making''.
\newblock {\em Sequence learning: paradigms, algorithms, and
  applications\/}:264--287.


\bibitem[\protect\citeauthoryear{Da~Silva, Basso, Bazzan, and Engel}{Da~Silva
  et~al.}{2006}]{da2006dealing}
Da~Silva, B.~C., E.~W. Basso, A.~L. Bazzan, and P.~M. Engel. 2006.
\newblock ``Dealing with non-stationary environments using context detection''.
\newblock In {\em Proceedings of the 23rd international conference on Machine
  learning},  217--224.

\bibitem[\protect\citeauthoryear{Dahlin, Bose, and Veeravalli}{Dahlin
  et~al.}{2022}]{dahlin2022controlling}
Dahlin, N., S.~Bose, and V.~V. Veeravalli. 2022.
\newblock ``Controlling a Markov Decision Process with an Abrupt Change in the
  Transition Kernel''.
\newblock {\em arXiv preprint arXiv:2210.04098\/}.


\bibitem[\protect\citeauthoryear{Dayan and Sejnowski}{Dayan and
  Sejnowski}{1996}]{dayan1996exploration}
Dayan, P., and T.~J. Sejnowski. 1996.
\newblock ``Exploration bonuses and dual control''.
\newblock {\em Machine Learning\/}~25:5--22.


\bibitem[\protect\citeauthoryear{Doya, Samejima, Katagiri, and Kawato}{Doya
  et~al.}{2002}]{doya2002multiple}
Doya, K., K.~Samejima, K.-i. Katagiri, and M.~Kawato. 2002.
\newblock ``Multiple model-based reinforcement learning''.
\newblock {\em Neural computation\/}~14(6):1347--1369.


\bibitem[\protect\citeauthoryear{Guan, Li, Duan, Wang, and Cheng}{Guan
  et~al.}{2018}]{Guan:2018}
Guan, Y., S.~E. Li, J.~Duan, W.~Wang, and B.~Cheng. 2018.
\newblock ``Markov probabilistic decision making of self-driving cars in
  highway with random traffic flow: a simulation study''.
\newblock {\em Journal of Intelligent and Connected Vehicles\/}~1(2):77--84.


\bibitem[\protect\citeauthoryear{Hadoux, Beynier, and Weng}{Hadoux
  et~al.}{2014}]{Hadoux:2014}
Hadoux, Beynier, and Weng. 2014.
\newblock ``Sequential decision-making under non-stationary environments via
  sequential change-point detection''.
\newblock {\em Learning over multiple contexts (LMCE)\/}.


\bibitem[\protect\citeauthoryear{Kushner and Yin}{Kushner and
  Yin}{1997}]{harold1997stochastic}
Kushner, H., and G.~Yin. 1997.
\newblock {\em Stochastic approximation and recursive algorithm and
  applications}, Volume~35.


\bibitem[\protect\citeauthoryear{Lai}{Lai}{1998}]{lai1998information}
Lai, T.~L. 1998.
\newblock ``Information bounds and quick detection of parameter changes in
  stochastic systems''.
\newblock {\em IEEE Transactions on Information theory\/}~44(7):2917--2929.


\bibitem[\protect\citeauthoryear{Meyn}{Meyn}{2022}]{Meyn:2022}
Meyn, S. 2022.
\newblock {\em Control systems and reinforcement learning}.
\newblock Cambridge University Press.


\bibitem[\protect\citeauthoryear{Sutton and Barto}{Sutton and
  Barto}{2018}]{Sutton:2018}
Sutton, R.~S., and A.~G. Barto. 2018.
\newblock {\em Reinforcement learning: An introduction}.
\newblock MIT press.


\bibitem[\protect\citeauthoryear{Tartakovsky, Nikiforov, and
  Basseville}{Tartakovsky et~al.}{2014}]{tartakovsky2014sequential}
Tartakovsky, A., I.~Nikiforov, and M.~Basseville. 2014.
\newblock {\em Sequential analysis: Hypothesis testing and changepoint
  detection}.
\newblock CRC Press.


\bibitem[\protect\citeauthoryear{Veeravalli and Banerjee}{Veeravalli and
  Banerjee}{2014}]{Veeravalli:2014}
Veeravalli, V.~V., and T.~Banerjee. 2014.
\newblock ``Quickest change detection''.
\newblock In {\em Academic press library in signal processing}, Volume~3,
  209--255. Elsevier.

\bibitem[\protect\citeauthoryear{Watkins and Dayan}{Watkins and
  Dayan}{1992}]{Watkins:1992}
Watkins, C.~J., and P.~Dayan. 1992.
\newblock ``Q-learning''.
\newblock {\em Machine learning\/}~8:279--292.


\end{thebibliography}


\end{document}